\documentclass[twocolumn]{autart}    

\usepackage{graphicx}          
\usepackage{graphics,color}
\usepackage{amssymb,amsmath,amsfonts}
\usepackage{amsmath,mathrsfs,bm,url,times}
\usepackage{latexsym}
\usepackage{subfigure}
\usepackage{algorithmic}
\usepackage{algorithm}
\usepackage{epsfig,url}

\usepackage[]{natbib}

\newtheorem{theorem}{Theorem}
\newtheorem{lemma}{Lemma}

\newtheorem{remark}{Remark}
\newtheorem{corollary}{Corollary}
\newtheorem{property}{Property}
\newenvironment{proof}{\vspace{1ex}\noindent{\bf Proof.}\hspace{0.5em}}
    {\hfill\qed\vspace{1ex}}
\def\qed{\hfill \vrule height 6pt width 6pt depth 0pt}

\newcommand{\onetom}{1,\cdots,m}

\begin{document}

\begin{frontmatter}

\title{Distributed Event-triggered Consensus for Multi-agent Systems with Directed Topologies\thanksref{footnoteinfo}} 

\thanks[footnoteinfo]{This paper was not presented at any IFAC
meeting. Corresponding author W.~L.~Lu. Tel. +86-21-65642903.
Fax +86-21-56521137.}

\author[First]{Xinlei Yi}\ead{yix11@fudan.edu.cn},    
\author[First,Second]{Wenlian Lu}\ead{wenlian@fudan.edu.cn},               
\author[First,Third]{Tianping Chen}\ead{tchen@fudan.edu.cn}  

\address[First]{School of Mathematical Sciences, Fudan University, Shanghai 200433, China}
\address[Second]{Centre for Computational Systems Biology, Fudan University, Shanghai 200433, China }
\address[Third]{School of Computer Science, Fudan University, Shanghai 200433, China}


\begin{abstract}                          
In this paper, we study consensus problem of multi-agent system in networks with directed topology by event-triggered feedback control. We derive distributed criteria to determine the next observation time of each agent that are triggered by its in-neighbours' information and its own states respectively. We prove that if the network topology has spanning tree, then under the event-triggered principles, the multi-agent system reaches consensus with the consensus value equal to a weighted average of all agents' initial values. These principles, does not need any a priori knowledge of any global parameter for each agent, and Zeno Behaviours can be excluded in the last two principles. In addition, these results are extended to the case of self-triggered control, in terms of that the next triggering time of each agent is predicted based on the states at the last event time, which implies that the system states are not be monitored in a simultaneous way. The effectiveness of the theoretical results is illustrated by numerical examples.
\end{abstract}

\end{frontmatter}

\section{Introduction}
In the past decades, consensus problem in multi-agent systems, that a group of agents seeks to agree upon certain quantity of interest, has attracted increasing research interests. Please see \cite{Ros,Liubo} and the reference therein, which proved that the sufficient condition for consensus is connectivity of the undirected graph or possessing a spanning tree of of the directed graph for the networked system. However, most of these papers assumed the continuous feedback of states as controller. Motivated by the future trend that agent can be equipped with embedded microprocessors with limited resources to transmit and collect information, event-triggered control was introduced by \cite{Pta,Mmp,Mmc, Mmc1,Ptn,Crf} and self-triggered control was studied by \cite{Aapt,Xwmd,Mjm,Aap}. Compared with continuous state feedback, event-triggered control strategy induces staircase control signal, which is constant between the so-called {\em trigger times}, dependent or independent of the system states. In addition, self-triggered control is an extension of the event-triggered control. That is, each agent predicts the next trigger time by discrete observation of states. \cite{Dvd,Gss,Mxy,Gar,Yfg,Now} considered the consensus problem for multi-agent systems with event-triggered control. In particular, \cite{Dvd} provided event-triggered and self-triggered approaches in both centralized and distributed formulations. for networked system with {\em undirected} graph topology.

Even-triggering algorithms for consensus can be regarded to linearize and discrete original algorithm. following two papers \cite{LC2004,LC2007} are closely relating to this issue. In particular, the self-triggered algorithms.

More related to the present work, \cite{Zliu} studied event-triggered principle for consensus of multi-agent system with directed and weighted but {\em balanced} graph. \cite{Zzf} proposed event-triggered principle for directed unbalanced graph topology, which, however, requested a priori knowledge of some global parameters, for instance, the Laplacian, to determine the next trigger time for each agent.

In this paper, we study event-triggered and self-triggered principles for consensus in multi-agent system with directed and weighted topology of possibly reducible and unbalanced Laplacian. In comparison to the literature, there are three main contributions: (i) the directed graph topologies are not necessarily balanced; (ii) all the event-triggered principles in this paper are essentially distributed, dependent or independent of the in-neighbours' states but without any priori global knowledge, and the exclusion of Zeno behaviour is proved for two of the three principles; (iii) we propose self-triggered principles, which implies economic communication load in comparison to  simultaneously monitoring the system states.


$\|\cdot\|$ represents the Euclidean $L_2$-norm of vectors or the induced $L_2$-norm of square matrices. $\bf 1$ denotes the column vector with each component equal to $1$ of an appropriate dimension.  $\rho(\cdot)$ stands for the spectral radius of symmetry matrices and $\rho_2(\cdot)$ is its minimum positive eigenvalue. 

\section{Preliminaries and problem formulation}
In this section, we present some definitions in algebraic graph theory (please see \cite{Die} and \cite{Rah} for details) and the formulation of the problem.

A weighted directed graph is defined by $\mathcal G=(\mathcal V, \mathcal E, \mathcal A)$ of $m$ agents (or nodes), with the node (agent) set $\mathcal V =\{v_1,\cdots,v_m\}$, the link (edge) set $\mathcal E \subseteq \mathcal V \times \mathcal V$, and the weight matrix $\mathcal A =[a_{ij}]_{i,j\in\mathcal I}$ with $a_{ij}>0$ the weight of the link $e(i,j)=(v_i, v_j)\in \mathcal E$ if there is a directed link from agent $v_{j}$ to agent $v_{i}$, and $a_{ij}=0$ otherwise. We take $a_{ii}=0$ for all $i\in \mathcal I$, where $\mathcal I=\{1, 2,\cdots, m\}$. The in- and out- neighbours set of agent $v_i$ are defined as $N^{in}_i=\{v_j\in \mathcal V\mid a_{ij}>0\},~~~N^{out}_i=\{v_j\in \mathcal V\mid a_{ji}>0\}$, the (weighted) in-degree of agent $v_i$ as $deg^{in}(v_i)=\sum_{j=1}^{m}a_{ij}$, and thus the (in-)degree matrix of $\mathcal G$  as $D=diag[deg^{in}(v_1), \cdots, deg^{in}(v_m)]$. We also define the weighted Laplacian matrix associated with the digraph $\mathcal G$ as $L=\mathcal A-D$. 
It is know that $\mathcal G$ being strongly connected is equivalent to the corresponding Laplacian matrix $L$ being irreducible. Furthermore, by Lemma 1 in \cite{LC2004a}, 
and Lemma 2.6 in \cite{Ren}, we have the following result.
\begin{lemma}\label{lem2}
(i) If $L$ is irreducible, then $rank(L)=m-1$; and zero is an algebraically simple eigenvalue of $L$ with a positive vector $\xi^{\top}=[\xi_{1},\cdots,\xi_{m}]$ satisfying $\xi^{\top} L=0$ and $\sum_{i=1}^{m}\xi_{i}=1$. 
Moreover, let $\Xi=diag[\xi_{1},\cdots,\xi_{m}]$, then $R=\frac{1}{2}(\Xi L+L^{\top}\Xi)$ is a symmetric matrix with all row sums equal to zeros and has zero eigenvalue with algebraic dimension one.
\end{lemma}

It can be seen that $R$ is negative semi-definite. Let $0=\lambda_{1}<\lambda_{2}\le\cdots\le\lambda_{m}$ be the eigenvalues of $-R$, counting the multiplicities. Let $U=\Xi-\xi\xi^{\top}$. Thus, $U$ has a single zero eigenvalue and its eigenvalues (counting
the multiplicities) can be sorted as $0=\mu_{1}<\mu_{2}\le\cdots\le\mu_{m}$. Then, it holds that $\lambda_{m}x^{\top}x\ge\min_{x\bot \bf 1}\{x^{\top}(-R)x\}\ge\lambda_{2}x^{\top}x,$ and $\mu_{m}x^{\top}x\ge\min_{x\bot \bf 1}\{x^{\top}Ux\}\ge\mu_{2}x^{\top}x.$ Therefore, $-R\ge \frac{\lambda_2}{\mu_m}U$ holds, i.e., $-R- \frac{\lambda_2}{\mu_m}U$ is positive semi-definite.

Consider the following multi-agent system with discontinuous diffusions:
\begin{eqnarray}
\begin{cases}
\dot{x}_{i}(t)=u_{i}(t)\\
u_{i}(t)=\sum_{j=1}^{m}L_{ij}x_{j}(t_{k_{j}(t)}^{j}),~i=\onetom\end{cases}\label{mg}
\end{eqnarray}
where $k_{i}(t)=arg\max_{k}\{t^{i}_{k}\le t\}$, the increasing time agent-wise sequence $\{t_{k}^{j}\}_{k=1}^{\infty}$, $j=\onetom$, named {\em trigger time}. We say that agent $v_i$ is triggered at $t=t^{i}_{k}$ if $v_{i}$ renews its state at $t=t^{i}_{k}$ and sends this renewed state $x_i(t^{i}_{k})$ to all its out-neighbours at the same time. Let $x(t)=[x_{1}(t),\cdots,x_{m}(t)]^{\top}$ and  $\hat{x}_{i}(t)=x_{i}(t_{k_{i}(t)}^{i})$,
$\hat{x}(t)=[\hat{x}_{1}(t), \cdots,\hat{x}_{m}(t)]^{\top}$,  $ e_{i}(t)=\hat{x}_{i}(t)-x_{i}(t)$ and $e(t)=[e_{1}(t),\cdots,e_{m}(t)]^{\top}$, $q_{i}(t)=\sum_{j=1}^{m}L_{ij}[x_{j}(t)-x_i(t)]^2$ and $\hat{q}_{i}(t)=\sum_{j=1}^{m}L_{ij}[\hat{x}_{j}(t)-\hat{x}_{i}(t)]^2$. $\bar{x}(t)=\sum_{j=1}^{m}\xi_{j}x_{j}(t)$.

\section{Event-triggered principles}\label{push}

In this section, we study event-triggered control for multi-agent
systems with directed and weighted topology.

Consider a candidate Lyapunov function (see \cite{LC2006}):
\begin{align}
&V(t)=\frac{1}{2}\sum_{i=1}^{m}\xi_{i}(x_{i}(t)-\bar{x}(t))^{2}
\label{V}
\end{align}
Then, the derivative of $V(t)$ along (\ref{mg}) gives
\begin{align}
\frac{d}{dt}V(t)=&\sum_{i=1}^{m}\xi_{i}x_{i}(t)
\sum_{j=1}^{m}L_{ij}x_{j}(t^{j}_{k_{j}(t)})
\nonumber\\
=&\sum_{i=1}^{m}\sum_{j=1}^{m}\xi_{i}L_{ij}x_{i}(t)x_{j}(t)\nonumber\\
&+\sum_{i=1}^{m}\sum_{j=1}^{m}\xi_{i}L_{ij}x_{i}(t)[x_{j}(t^{j}_{k_{j}(t)})
-x_{j}(t)]
\nonumber\\
=&-\frac{1}{2}\sum_{i=1}^{m}\xi_{i}q_i(t)
-\sum_{i=1}^{m}\sum_{j\neq i}^{m}\xi_{i}L_{ij}[x_{j}(t)-x_i(t)]e_j(t)\nonumber\\
\le&-\frac{1}{2}\sum_{i=1}^{m}\xi_{i}q_i(t)
+\sum_{i=1}^{m}\sum_{j\neq i}^{m}\xi_{i}L_{ij}\Big\{\frac{1}{4}[x_{j}(t)-x_i(t)]^2\nonumber\\
+&[e_j(t)]^2\Big\}
=-\frac{1}{4}\sum_{i=1}^{m}\xi_{i}q_i(t)+\sum_{i=1}^{m}\xi_i|L_{ii}|[e_i(t)]^2
\label{dV1.1}
\end{align}

Then,  we have the following result.
\begin{theorem}\label{coro1.1}
Suppose that $\mathcal G$ is strongly connected. For agent $v_i$, if trigger times $t^{i}_{1},\cdots,t^{i}_{k}$ are given, then use the following trigger strategy to find $t^{i}_{k+1}$:
\begin{align}
t_{k+1}^{i}=\max\big\{&\tau\ge t_{k}^{i}:~|x_{i}(t^{i}_{k})-x_{i}(t)|\nonumber\\
&\le \sqrt{\frac{\gamma_{i}}{4|L_{ii}|}q_{i}(t)},\forall t\in[t_{k}^{i},\tau]\big\}\label{event1.1c}
\end{align}
with $\gamma_{i}\in(0,1)$.
Then, if $\lim_{k\to\infty}t^{i}_{k}=+\infty$, then system (\ref{mg}) reaches consensus  exponentially; In addition, $\lim_{t\to\infty}x_{i}(t)=\sum_{j=1}^{m}\xi_{j}x_{j}(0)$ for all $i\in\mathcal I$.
\end{theorem}
\begin{proof}
From inequality (\ref{dV1.1}) and condition (\ref{event1.1c}),
we have
\begin{align*}
\frac{d}{dt}V(t)\le&-\sum_{i=1}^{m}\frac{1-\gamma_{i}}{4}\xi_iq_{i}(t)
\le-\frac{1-\gamma}{2}x^{\top}(t)(-R)x(t)\\
\le&-\frac{(1-\gamma)\lambda_2}{\mu_m}V(t)
\end{align*}
for all $t\ge 0$, where $\gamma=max\{\gamma_1,\cdots,\gamma_m\}$, which means $
V(t)\le e^{-[(1-\gamma)\lambda_2/\mu_m] t}V(0)$.
This implies that system (\ref{mg}) reaches consensus exponentially. Combined with $\bar{x}(t)=\bar{x}(0)$ for all $t\ge 0$,  we have $\lim_{t\to\infty}x_{i}(t)=\sum_{j=1}^{m}\xi_{j}x_{j}(0)$ for all $i\in\mathcal I$.
\end{proof}
\begin{remark}
The principle in Theorem \ref{coro1.1} indicates that each needs only its in-neighbours' state information, without any priori knowledge of any global parameter.
\end{remark}
\begin{property}\label{psh}
Under the condition in Theorem \ref{coro1.1}, for any agent $v_i$, if $x_{i}(t^{i}_k)\neq\bar{x}(0)$, then $t^{i}_{k+1}>t^{i}_{k}$ holds.
\end{property}
In fact, suppose that $t^{i}_k$ is agent $v_i$'s last trigger time. Then, from (\ref{event1.1c}), we have $x_{i}(t^{i}_k)=\lim_{t\to\infty}x_{i}(t)=\bar{x}(0)$. However, this could not happen since $x_{i}(t^{i}_k)\neq\bar{x}(0)$.

However, we cannot exclude Zeno behaviour for the event-triggered principle in Theorem \ref{coro1.1}. To out best knowledge, this problem has not been solved yet. Zeno behaviour is defined as that there are infinite number of triggers in a finite time interval (\cite{Joh}).
Inspired by \cite{Gss}, we give another event-triggered principle without Zeno behaviour.
\begin{theorem}\label{coro1.1e}
Suppose that $\mathcal G$ is strongly connected. For agent $v_i$, if trigger times $t^{i}_{1},\cdots,t^{i}_{k}$ are given, then use the following trigger strategy to find $t^{i}_{k+1}$:
\begin{align}
t_{k+1}^{i}=\max\Big\{\tau\ge t_{k}^{i}:~|x_{i}(t^{i}_{k})-x_{i}(t)|\le \delta_i(t),~\forall t\in[t_{k}^{j},\tau]\Big\}\label{event1.1e}
\end{align}
where $\delta_i(t)=\phi_i e^{-\alpha_it}$, with $\phi_i>0$, $\alpha_i>0$.  Then, system (\ref{mg}) reaches consensus  exponentially  and Zeno behaviour is excluded; $\lim_{t\to\infty}x_{i}(t)=\sum_{j=1}^{m}\xi_{j}x_{j}(0)$ for all $i\in\mathcal I$.
\end{theorem}
\begin{proof}
Direct calculations give
\begin{align}
&\frac{d}{dt}V(t)=x^{\top}(t)Rx(t)+x^{\top}(t)\Xi Le(t)\nonumber\\
\le& x^{\top}(t)Rx(t)+\frac{a}{2}x^{\top}(t)\Xi LL^{\top}\Xi x(t)+\frac{1}{2a}e^{\top}(t)e(t)\nonumber\\
\le&-[1-\frac{a\rho(Q)}{2\lambda_2}]x^{\top}(t)(-R)x(t)+\frac{1}{2a}e^{\top}(t)e(t)\label{dV1}
\end{align}
with any $a>0$ and $Q=\Xi LL^{\top}\Xi$,

Let $\phi=\max\{\phi_1,\cdots,\phi_m\}$, $\alpha=\min\{\alpha_1,\cdots,\alpha_m,\frac{\lambda_2}{2\mu_m}\}$, and then $\delta(t)=\phi e^{-\alpha t}$.
Let $a=\lambda_2/\rho(\Xi LL^{\top}\Xi)$ in (\ref{dV1}). Then under the condition (\ref{event1.1e}),
$dV(t)/dt\le-(\lambda_2/\mu_m)V(t)+[m/(2a)]\delta^2(t)$ holds. By the Gr\"{o}nwell inequality, we have
\begin{align}
&V(t)\le e^{-ct}V(0)+\frac{m}{2a}\int_{0}^{t}e^{-c(t-s)}\delta^2(s)ds\nonumber\\
=&\begin{cases}
e^{-ct}V(0)+\frac{m\phi^2}{2a(c-2\alpha)}[e^{-2\alpha t}-e^{-ct}],~if~\alpha<\frac{c}{2}\\
e^{-ct}V(0)+\frac{m}{2a}e^{-ct}t,~if~\alpha=\frac{c}{2}
\end{cases}\nonumber\\
\le& k_{\delta}e^{-\alpha t}\label{Ve}
\end{align}
with $c=\frac{\lambda_2}{\mu_m}$ and $k_{\delta}=\begin{cases}V(0)+\frac{m}{2a(c-2\alpha)}>0,~if~\alpha<\frac{c}{2}\\
V(0)+\frac{m}{2a}T_0>0,~if~\alpha=\frac{c}{2}\end{cases}$ is a constant, where $T_0>0$ with $e^{-\alpha t}t\le1$ for all $t\ge T_0$. This implies that system (\ref{mg}) reaches consensus  exponentially with $\lim_{t\to\infty}x_{i}(t)=\sum_{j=1}^{m}\xi_{j}x_{j}(0)$ for all $i\in\mathcal I$.

It can be seen that $\|x(t)\|$ is bounded on $[0,+\infty)$, namely there exists $M>0$ such that $|\sum_{j=1}^{m}L_{ij}\hat{x}_j(t)|\le M$, $\forall i\in\mathcal I$ and $\forall t\ge0$. Thus
$|x_{i}(t)-x_{i}(t^{i}_{k_{i}(t)})|\le M(t-t^{i}_{k_{i}(t)})$.

Then, to prove Zeno behaviour exclusion, let $t^{i}_k=t^{i}_{k_{i}(t)}$ for all $t\in[0,T]$, which implies $
\delta_i(t)\ge\phi_i e^{-\alpha_iT}e^{-\alpha_i(t-t^{i}_k)}\ge\phi_i e^{-\alpha_iT}[1-\alpha_i(t-t^{i}_k)]$.
For agent $v_i$, when the event is triggered, i.e., the equality of (\ref{event1.1e}) holds at $t=t^{i}_{k+1}$, which implies $|x_{i}(t^{i}_{k+1})-x_{i}(t^{i}_{k})|=\delta_i(t^{i}_{k+1})$, namely,
\begin{align}
M(t^{i}_{k+1}-t^{i}_{k})\ge\phi_i e^{-\alpha_iT}[1-\alpha_i(t^{i}_{k+1}-t^{i}_{k})]\label{event1.1esss}
\end{align}
Hence, every inter-event time is lower bounded by $\frac{1}{\frac{M}{\phi_i}e^{\alpha_iT}+\alpha_i}$. Therefore, for agent $v_i$, there exist at most $N=\lceil(\frac{M}{\phi_i}e^{\alpha_iT}+\alpha_i)T+1\rceil$ triggers in $[0,T]$, where $\lceil z\rceil$ stands for the largest integer less than $z$.
\end{proof}
\begin{remark}
(i) Compared with the work \cite{Gss}, the trigger times of each agent determined by Theorem \ref{coro1.1e} only depends on its own state information, and is independent of any priori global knowledge or its neighbours'. (ii) If $\phi_1=\cdots=\phi_m=\delta_0>0$ and $0<\alpha_1=\cdots=\alpha_m=\delta_1<\frac{\lambda_2}{2\mu_m}$, then similar to the proof of Theorem 3.2 in \cite{Gss}, we have $t^i_{k+1}-t^i_{k}\ge\frac{1}{\omega_{i}+\delta_1}>0$, where $\omega_{i}=2|L_{ii}|+\sqrt{\frac{\rho(L^{\top}L)}{\mu_2}k_{\delta}}$ and $k_{\delta}=\frac{V(0)}{\delta^2_0}+\frac{m}{2a(c-2\delta_1)}>0$. So, the lower-bound of the inter event intervals depends of the ``global'' information of the Laplacian. However, for the triggering principle, we do not need this global information.
\end{remark}


Following event-triggered principle (also see \cite{Gar} and \cite{Now}) can exclude Zeno behaviour and without knowing any priori global knowledge.
\begin{theorem}\label{thmabc}
Suppose that $\mathcal G$  is strongly connected. For agent $v_i$, pick $\gamma_{i}\in(0,1)$ and $0<\varepsilon_i<\sqrt{\gamma_i}/(2|L_{ii}|)$. With $t^{i}_{1},\cdots,t^{i}_{k}$, the following triggered strategy is used to find $t^{i}_{k+1}$:
\begin{enumerate}
\item Obtain $\tau^{i}_{k+1}$:
\begin{align}
\tau_{k+1}^{i}=\max\big\{\tau\ge t_{k}^{i}:~g_i(e_i(t))\le0,\forall t\in[t_{k}^{i},\tau]\big\}\label{event1.1cr}
\end{align}
where $g_i(e_i(t))=|e_i(t)|-\sqrt{\frac{\gamma_{i}}{4|L_{ii}|}\hat{q}_{i}(t)}$;
\item If agent $v_i$ receives renewed information from its in-neighbours in $(t^{i}_{k},\tau_{k+1}^{i}]$, letting  $t_0$ be the first time when agent $v_i$ receives renewed information from its in-neighbours with $t_0-t^{i}_{k}<\varepsilon_i$, then a$t^{i}_{k+1}=t_0$; Otherwise $t^{i}_{k+1}=\tau^{i}_{k+1}$.
\end{enumerate}
Then, system (\ref{mg}) reaches consensus exponentially excluding the Zeno behaviour\footnote{It could happen that $t^{i}_{k+1}$ cannot be found by Theorem \ref{thmabc}, i.e $t^{i}_{k}$ is agent $v_i$'s last trigger time and it does not trigger after $t^{i}_{k}$ any more. If so, obviously, for agent $v_i$, we can conclude that there do not exist an infinite number of triggers in a finite time period. Without loss of generality, we assume that $t^{i}_{l},l=1,2,\cdots$ exist.} and with $\lim_{t\to\infty}x_{i}(t)=\sum_{j=1}^{m}\xi_{j}x_{j}(0)$ for all $i\in\mathcal I$.
\end{theorem}
\begin{proof}
It also can be written as
\begin{align}
\frac{d}{dt}V(t)
=&-\frac{1}{2}\sum_{i=1}^{m}\xi_{i}\hat{q}_{i}(t)
-\sum_{i=1}^{m}\sum_{j\neq i}^{m}\xi_{i}L_{ij}e_{i}(t)[\hat{x}_{j}(t)-\hat{x}_{i}(t)]\nonumber\\
\le&-\frac{1}{4}\sum_{i=1}^{m}\xi_{i}\hat{q}_{i}(t)
+\sum_{i=1}^{m}|L_{ii}|e_{i}^2(t)\nonumber\\
\le&-\frac{1}{4}\sum_{i=1}^{m}(1-\gamma_i)\xi_{i}\hat{q}_{i}(t)\nonumber\\
\le&\frac{1-\gamma}{2}[\hat{x}(t)]^{\top}R\hat{x}(t)
\label{dV1.2}
\end{align}
where $\gamma=\max\{\gamma_1,\cdots,\gamma_m\}$. 
In addition,
\begin{align*}
V(t)=&\frac{1}{2}x^{\top}(t)Ux(t)
\le\frac{\mu_m}{2\lambda_2}x^{\top}(t)(-R)Rx(t)\nonumber\\
=&\frac{\mu_m}{2\lambda_2}[\hat{x}(t)-e(t)]^{\top}(-R)[\hat{x}(t)-e(t)]\nonumber\\
\le&\frac{\mu_m}{\lambda_2}\Big\{[\hat{x}(t)]^{\top}(-R)\hat{x}(t)
+[e(t)]^{\top}(-R)e(t)\Big\}
\end{align*}
and combine (\ref{event1.1cr}), we have
\begin{align*}
[e(t)]^{\top}(-R)e(t)\le&\lambda_m\|e(t)\|^2\nonumber\\
\le&\frac{\lambda_m\gamma}{(2\min_{i}\{|L_{ii}|\xi_i\})}[\hat{x}(t)]^{\top}(-R)\hat{x}(t)
\end{align*}
So, we have
$dV(t)/dt\le-(1/k_v)V(t)$
where $k_v=\frac{2\mu_m}{\lambda_2(1-\gamma)}[1+\frac{\lambda_m\gamma}{2\min_{i}\{|L_{ii}|\xi_i\}}]>0$. Therefore, system (\ref{mg}) reaches consensus exponentially and $\lim_{t\to\infty}x_{i}(t)=\sum_{j=1}^{m}\xi_{j}x_{j}(0)$ can be derived by the same way above.

Next, we prove that the Zeno behaviour can be excluded. To start with, we will prove that, for agent $v_i$, under the condition that trigger times $t^{i}_{1}=0,\cdots,t^{i}_{k}$ are given, if it does not receive new state information from its in-neighbours after $t^{i}_k$, then the lower bound of $t^{i}_{k+1}-t^{i}_k$ is $\frac{\sqrt{\gamma_i}}{2|L_{ii}|}$. In this case, we have $e_{i}(t^{i}_{k})=0$ and $\sum_{j=1}^{m}L_{ij}\hat{x}_j(t)=\sum_{j=1}^{m}L_{ij}\hat{x}_j(t^{i}_{k})$ since no new information is received for $t\ge t^{i}_{k}$, which implies
$e_{i}(t)=-(t-t^{i}_{k})\sum_{j=1}^{m}L_{ij}\hat{x}_j(t^{i}_{k})$.

Note that if $\sum_{j=1}^{m}L_{ij}\hat{x}_j(t^{i}_{k})=0$, no trigger will ever happen since $e_{i}(t)=0$ for all $t\ge t^{i}_{k}$. If $\sum_{j=1}^{m}L_{ij}\hat{x}_j(t^{i}_{k})\neq0$ (thus $\hat{q}_{i}(t)=\hat{q}_{i}(t^{i}_{k})\neq0$, for $t>t^{i}_{k}$), the event (\ref{event1.1cr}) prescribes a trigger at the time $t^{*}\ge t^{i}_{k}$ satisfying$
(t^{*}-t^{i}_{k})^2\Big[\sum_{j=1}^{m}L_{ij}\hat{x}_j(t^{i}_{k})\Big]^2-\frac{\gamma_{i}}{4|L_{ii}|}\hat{q}_{i}(t^{i}_{k})=0.
$
Thus, we have
\begin{align}
&t^{*}-t^{i}_{k}
=\sqrt{\frac{\gamma_{i}}{4|L_{ii}|}\frac{\sum_{j\neq i}^{m}L_{ij}[\hat{x}_{j}(t^{i}_{k})-\hat{x}_{i}(t^{i}_{k})]^2}{\Big[\sum_{j\neq i}^{m}L_{ij}(\hat{x}_j(t^{i}_{k})-\hat{x}_i(t^{i}_{k}))\Big]^2}}\nonumber\\
\ge&\sqrt{\frac{\gamma_{i}}{4|L_{ii}|}\frac{\sum_{j\neq i}^{m}L_{ij}[\hat{x}_{j}(t^{i}_{k})-\hat{x}_{i}(t^{i}_{k})]^2}
{\sum_{j\neq i}^{m}L_{ij}\sum_{j\neq i}^{m}L_{ij}[\hat{x}_{j}(t^{i}_{k})-\hat{x}_{i}(t^{i}_{k})]^2}}\nonumber\\
=&\frac{\sqrt{\gamma_i}}{2|L_{ii}|}\label{ietl}
\end{align}

Let $\varepsilon_0=min\{\varepsilon_1,\cdots,\varepsilon_m\}$. We will prove that, for agent $v_i$, in any time interval with length of $\frac{1}{2}\varepsilon_0$, i.e., $\mathcal J=[T,T+\frac{1}{2}\varepsilon_0]$ for any $T\ge0$, there exist at most $m$ triggers.
Without loss of generality, let $k_{0}$ be a positive integer which satisfies $t^{i}_{k_0-1}<T$ and $t^{i}_{k_0}\ge T$. In fact, if there are no such $k_{0}$, then $\mathcal J$ has no triggers.

First, if no information is received during the time period $(t^{i}_{k_{0}},t^{i}_{k_{0}}+\varepsilon_i)$, then from (\ref{ietl}), it holds that $t^{i}_{k_{0}+1}-t^{i}_{k_{0}}\ge\varepsilon_i>\varepsilon_{0}$. Hence, there is only one trigger in $\mathcal J$. Otherwise, if at least one in-neighbour of agent $v_{i}$ triggers at some time in $\mathcal J$, then letting $t_1\in(t^{i}_{k},t^{i}_{k}+\varepsilon_i)$ be the first time after $t^{i}_{k_{0}}$ that agent $v_{i}$ receives trigger information from its in-neighbours. Then, $t^{i}_{k+1}=t_1$ according to the second item of the trigger rule. There exists a nonempty set, denoted by $\mathcal{I}_{1}$, of all agents who trigger at time $t_1$.

Again, if there are more triggers in $\mathcal J$. Let $t_{2}$ be the next trigger time after $t_{1}$. We claim that there must be a nonempty set of agents, denoted by $\mathcal I_{2}$, which trigger at $t_{1}$, such that $\mathcal I_{2}\bigcup\mathcal I_{1}-\mathcal I_{1}\ne\emptyset$. In fact, if not so, $t_{2}-t_{1}\ge\epsilon_{0}$ holds, which implies $t_{2}\notin\mathcal J$. That is, $\#(\mathcal I_{2}\bigcup\mathcal I_{1})\ge 2$.

Reasoning repetitively in this way, let $t_{l}$ be the $l+1$-th trigger time of agent $i$ with $t_{0}=t^{i}_{k}$ and  $\mathcal I_{l}$ be the set of agents that triggers at $t_{l}$. By induction, we have $\bigcup_{p=0}^{l-1}\mathcal{I}_{p}\ge l$. Since there is at most $m$ in the network, we can conclude that there are at most $m$ triggers in $\mathcal J$. This completes the proof.
\end{proof}

Inspired by \cite{LC2006}, previous theorems can be extended to the directed and reducible topology. Assume that $\mathcal G$ has a spanning tree. Without loss of generality, $L$ can be written in the
following Perron-Frobenius form:
\begin{eqnarray}
L=\left[\begin{array}{llll}L^{1,1}&L^{1,2}&\cdots&L^{1,K}\\
0&L^{2,2}&\cdots&L^{2,K}\\
\vdots&\vdots&\ddots&\vdots\\
0&0&\cdots&L^{K,K}
\end{array}\right]\label{PF}
\end{eqnarray}
where $L^{k,k}$, with dimension $n_{k}$, associated with the $k$-th strongly connected component (SCC) of $\mathcal G$, denoted by $SCC_{k}$, $k=1,\cdots,K$, and for each $k=1,\cdots,K-1$, there exists some $j>k$, such that $L^{k,j}\ne 0$.

For $SCC_{k}$, define $x^{k}(t)=[x_{1}^{k}(t),\cdots,x_{n_{k}}^{k}(t)]^{\top}$, $\hat{x}^{k}(t)=[\hat{x}_{1}^{k}(t),\cdots,\hat{x}_{n_{k}}^{k}(t)]^{\top}$, $e^{k}(t)=[e_{1}^{k}(t),\cdots,e_{n_{k}}^{k}(t)]^{\top}$, where $e^{k}_{i}(t)=\hat{x}_{i}^{k}(t)-x^{k}_{i}(t)$, and $\hat{x}_{i}^{k}(t)=x^{k}_{i}(t^{i+n_{k-1}}_{k_{i+n_{k-1}}(t)})$. For $v_i\in SCC_{k}$, denote $q^{k}_{i}(t)=\sum_{j=1}^{n_k}L^{k,k}_{ij}[x^{k}_{j}(t)-x^{k}_{i}(t)]^2$, $\hat{q}^{k}_{i}(t)=\sum_{j=1}^{n_k}L^{k,k}_{ij}[\hat{x}^{k}_{j}(t)-\hat{x}^{k}_{i}(t)]^2$,
$q^{k,K}_{i}(t)=\sum_{p=k+1}^{K}\sum_{j=1}^{n_p}L^{k,p}_{ij}[x^{p}_{j}(t)-x^{k}_{i}(t)]^2$, and $\hat{q}^{k,K}_{i}(t)=\sum_{p=k+1}^{K}\sum_{j=1}^{n_p}L^{k,p}_{ij}[\hat{x}^{p}_{j}(t)-\hat{x}^{k}_{i}(t)]^2$.
Define an auxiliary matrix $\tilde{L}^{k,k}=[\tilde{L}^{k,k}_{ij}]_{i,j=1}^{n_{k}}$ as\\
$
\tilde{L}^{k,k}_{ij}=\begin{cases}L^{k,k}_{ij}&i\ne j\\
-\sum_{p=1,p\not=i}^{n_{k}}L^{k,k}_{ip}&i=j\end{cases}
$,
then each $\tilde{L}^{k,k}$, $k=1,\cdots,K$, is irreducible or is of dimension one. Then, let $D^{k}=L^{k,k}-\tilde{L}^{k,k}=diag[D^{k}_{1},\cdots,D^{k}_{n_{k}}]$, which is a diagonal negative semi-definite matrix and has at least one diagonal negative (nonzero). Actually, $D^{k}_{i}=-\sum_{p=k+1}^{K}\sum_{j=1}^{n_p}L^{k,p}_{ij}$. Lemma \ref{lem2} implies that we can find  ${\xi^{k}}^{\top}$ to be the left eigenvector of $\tilde{L}^{k,k}$ corresponding to the eigenvalue zero and the sum of its components equals $1$. (Actually, $[0,\cdots,0,\xi^{K}_{1},\cdots,\xi^{K}_{n_K}]^{\top}$ is the nonnegative left eigenvector of $L$ corresponding to the eigenvalue zero and the sum of its components equals $1$.) Finally, let $\Xi^{k}=diag[\xi^{k}]$ $Q_{k}=\frac{1}{2}[\Xi^{k}L^{k,k}+(\Xi^{k}L^{k,k})^{\top}]$.
\begin{property}\label{p1}
Under the setup above, $Q_{k}$ is negative definite for all $k<K$.
\end{property}
Now we begin to extend the result in Theorem \ref{coro1.1} to the case of directed and reducible topology. For simplicity, hereby we only consider the case of $K=2$. The case $K>2$ can be treated in the same way.
First, according to Theorem \ref{coro1.1}, one can conclude that $x^{2}_{j}(t)$ and $\hat{x}_{j}^{2}(t)$, $j=1,\cdots,n_{2}$, converge to $\nu(0)$ exponentially, where $\nu(t)=\sum_{p=1}^{n_{2}}\xi^{2}_{p}x^{2}_{p}(t)$.

Then, for the subsystem of the $SCC_1$, defining $
V_{1}(t)=\frac{1}{2}\sum_{i=1}^{n_{1}}\xi^{1}_{i}[x^{1}_{i}(t)-\nu(t)]^{2}
$, we have
\begin{align}
&\frac{d}{dt}V_{1}(t)=\sum_{i=1}^{n_{1}}\xi^{1}_{i}[x^{1}_{i}(t)-\nu(t)]
\left\{\sum_{j=1}^{n_{1}}
L^{1,1}_{ij}[x^{1}_{j}(t)-\nu(t)]\right.\nonumber\\
&\left.+\sum_{p=1}^{n_{2}}L^{1,2}_{ip}
[\hat{x}^{2}_{p}(t)-\nu(t)]+\sum_{j=1}^{n_{1}}
L^{1,1}_{ij}[\hat{x}^{1}_{j}(t)-x^{1}_{j}(t)]\right\}\nonumber\\
&=Q^{1}_{1}(t)+Q^{1}_{2}(t)+Q^{1}_{3}(t)\label{dVK-1}
\end{align}
where
\begin{align*}
Q^{1}_{1}(t)=&[x^{1}(t)-\nu(t){\bf 1}]^{\top}\Xi^{1}L^{1,1}[x^{1}(t)-\nu(t){\bf 1}]\\
=&\sum_{i=1}^{n_1}\xi^{1}_{i}D^{1}_{i}[x^{1}_{i}(t)-\nu(t)]^2-\frac{1}{2}\sum_{i=1}^{n_1}\xi^{1}_{i}q^1_i(t)\\
Q_{2}^{1}(t)=&\sum_{i=1}^{n_{1}}\xi^{1}_{i}[x^{1}_{i}(t)-\nu(t)]\sum_{p=1}^{n_{2}}L^{1,2}_{ip}
[\hat{x}^{2}_{p}(t)-\nu(t)]\\
Q^{1}_{3}(t)=&[x^{1}(t)-\nu(t){\bf 1}]^{\top}\Xi^{1}L^{1,1}e^{1}(t)
\end{align*}
By the Cauchy inequality, for any $\kappa_1>0$, we have
\begin{align}
Q^{1}_{2}(t)\le\kappa_1 V_{1}(t)+F_{1}(t)\label{Q2}
\end{align}
where
$F_{1}(t)=\frac{1}{4\kappa_1}\sum_{i=1}^{n_{1}}\xi_{i}^{1}
\left\{\sum_{p=1}^{n_{2}}L^{1,2}_{ip}
[\hat{x}^{2}_{p}(t)-\nu(t)]\right\}^{2}$.
Since $\lim_{t\to\infty}\hat{x}^{2}_{p}(t)=\nu(0)=\nu(t)$, $p=1,\cdots,n_{2}$, hold exponentially,
$\lim_{t\to\infty}F_{1}(t)=0
$
exponentially.

For $Q^{1}_{3}(t)$, we have
\begin{align*}
&Q^{1}_{3}(t)=\sum_{i=1}^{n_1}\sum_{j=1}^{n_1}[x^{1}_{i}(t)-\nu(t)]\xi^{1}_{i}L^{1,1}_{ij}e^{1}_j(t)\\
=&\sum_{i=1}^{n_1}\xi^{1}_{i}D^{1}_{i}e^{1}_i(t)[x^{1}_{i}(t)-\nu(t)]\\
&-\sum_{i=1}^{n_1}\sum_{j\neq i}^{n_1}\xi^{1}_{i}\tilde{L}^{1,1}_{ij}e^{1}_j(t)[x^{1}_{j}(t)-x^{1}_{i}(t)]\\
\le&-\sum_{i=1}^{n_1}\xi^{1}_{i}D^{1}_{i}d^{1}_{i}[x^{1}_{i}(t)-\nu(t)]^2
-\sum_{i=1}^{n_1}\xi^{1}_{i}D^{1}_{i}\frac{1}{4d^{1}_{i}}[e^{1}_i(t)]^2\\
&+\sum_{i=1}^{n_1}\xi^{1}_{i}|\tilde{L}^{1,1}_{ii}|[e^{1}_i(t)]^2
+\sum_{i=1}^{n_1}\xi^{1}_{i}\frac{1}{4}q^{1}_{i}(t)
\end{align*}
with $d^{1}_{i}>0$.
Thus, we have
\begin{corollary}\label{coro2.1r}
Suppose that $\mathcal G$ has spanning tree and $L$ is in the form of (\ref{PF}) with $K=2$.
For $v_{p}\in SCC_{k}$ the trigger time sequence $\{t^{p+n_{k-1}}_{l}\}$ is determined by
\begin{align}
&t^{p+n_{k-1}}_{l+1}=\max\Big\{\tau\ge t^{p+n_{k-1}}_{l}:~|x_{p}^{k}(t)-x_{p}^{k}(t^{p+n_{k-1}}_{l})|\nonumber\\
&\le \sqrt{\frac{\gamma_{p}^{k}}{u_{p}^{k}}[q^{k}_{p}(t)+q^{k,K}_{p}(t)]},~\forall t\in[t^{p+n_{k-1}}_{l},\tau]\Big\}\label{event2.1r}
\end{align}
where $\gamma_{p}^{k}\in(0,1)$, $u_{p}^{k}=4|L_{pp}^{k,k}|-\frac{D_{p}^{k}}{d_{p}^{k}}$ and $0<d_{p}^{k}<\frac{1}{2}$.
Then, system (\ref{mg}) reaches consensus exponentially; In addition, $\lim_{t\to\infty}x_{i}(t)=\sum_{p=1}^{n_{2}}\xi^{2}_{p}x^{2}_{p}(0)$.
\end{corollary}
\begin{proof} We only need to discuss the components $v_{p}\in SCC_{1}$. From (\ref{event2.1r}) and (\ref{dVK-1}), we have
\begin{align*}
&\frac{d}{dt}V_{1}(t)
\le \sum_{i=1}^{n_1}(1-d^{1}_{i})\xi^{1}_{i}D^{1}_{i}[x^{1}_{i}(t)-\nu(t)]^2\\
&+\sum_{i=1}^{n_1}\gamma^{1}_{i}\xi^{1}_{i}\frac{1}{4}
\sum_{j=1}^{n_2}L^{1,2}_{ij}[x^{2}_{j}(t)-x^{1}_{i}(t)]^2\\
&-\frac{1}{4}\sum_{i=1}^{n_1}(1-\gamma^{1}_{i})\xi^{1}_{i}q^1_i(t)+\kappa_1 V_{1}(t)+F_{1}(t)\\
\le&\frac{1}{2}(1-\tilde{\gamma}_1)Q^{1}_{1}(t)+\kappa_1 V_{1}(t)+F_{1}(t)\\
&+\sum_{i=1}^{n_1}\gamma^{1}_{i}\xi^{1}_{i}[\frac{1}{4}+\frac{1}{16(\frac{3}{4}-d^{1}_{i})}]
\sum_{j=1}^{n_2}L^{1,2}_{ij}[x^{2}_{j}(t)-\nu(t)]^2\\
\le&-(1-\tilde{\gamma}_1)\frac{\rho(-Q_1)}{\rho_2(\Xi^1)}V_{1}(t)+\kappa_1 V_{1}(t)+F_{2}(t)
\end{align*}
where $F_{2}(t)=\sum_{i=1}^{n_1}\gamma^{1}_{i}\xi^{1}_{i}[\frac{1}{4}+\frac{1}{16(\frac{3}{4}-d^{1}_{i})}]
\sum_{j=1}^{n_2}L^{1,2}_{ij}[x^{2}_{j}(t)-\nu(t)]^2+F_{1}(t)$ and $\tilde{\gamma}_1=max\{\gamma^{1}_{1},\cdots,\gamma^{1}_{n_1}\}$. Picking sufficiently small $\kappa_1$, there exists some $\tilde{\kappa}_1>0$, such that
$
\frac{d}{dt}V_{1}(t)
\le-\tilde{\kappa}_1V_{1}(t)+F_{2}(t)
$, which implies
$
V_{1}(t)
\le e^{-\tilde{\kappa}_1t}\bigg\{V_{1}(0)+\int_{0}^{t}e^{\tilde{\kappa}_1s}F_{2}(s)ds\bigg\}
$. Noting $\lim_{t\to\infty}F_{2}(t)=0$ exponentially, since $\lim_{t\to\infty}F_{1}(t)=0$ and $\lim_{t\to\infty}\|x^2(t)-\nu(t){\bf 1}\|=0$ exponentially. Thus $\lim_{t\to\infty}V_{1}(t)=0$ exponentially.
\end{proof}

Letting $d_{p}^{k}=\frac{1}{4}$ in Corollary \ref{coro2.1r}, we can extend Theorems \ref{coro1.1}, Theorems \ref{coro1.1e} and \ref{thmabc} to the case that the graph has a spanning tree. Hence, we summarize the results as follows.
\begin{theorem}\label{coro1.1rr}
Suppose that $\mathcal G$ has spanning tree and $L$ is in the form of (\ref{PF}). For agent $v_i$, by the trigger rule as described in Theorem \ref{coro1.1}, assuming $\lim_{k\to\infty}t^{i}_{k}=+\infty$ for all $i=\onetom$, system (\ref{mg}) reaches consensus exponentially and $t^{i}_{k+1}>t^{i}_{k}$ holds whenever $x^{i}(t^{i}_{k})\ne\bar{x}(0)$; by the trigger rules as described in Theorems \ref{coro1.1e} and \ref{thmabc}, system (\ref{mg}) reaches consensus exponentially and the Zeno behaviour can be excluded. In addition, in all cases,
$\lim_{t\to\infty}x_{i}(t)=\sum_{p=1}^{n_{K}}\xi^{K}_{p}x^{K}_{p}(0)$.
\end{theorem}
It can be seen that Theorem \ref{coro1.1}, Theorem \ref{coro1.1e} and Theorem \ref{thmabc} are special cases of Theorem \ref{coro1.1rr}.

\section{Distributed self-triggered principles}
In the previous results, the continuous monitoring system states is required, which may cause cost communication load. An alternativeis to predict the next trigger time based on the states at the the agents' latest trigger time directly. 
This sort of triggered strategies are said to be {\em self-triggered principle}.

For any $p\in\mathcal I$,  $x_{p}(t)$ can be rewritten as 
\begin{eqnarray}
x_{p}(t)=x_{p}(t^{*}_{k_{p}(t)})
+(t-t^{*}_{k_{p}(t)})\sum_{j=1}^{m}L_{pj}x_j(t^{j}_{k_{j}(t)})\label{xpin}
\end{eqnarray}
for all $t>t^{*}_{k_{p}(t)}$ with
\begin{align}
t^{*}_{k_{p}(t)}=\max_{v_j\in  N^{in}_{p}}t^{j}_{k_{j}(t)}\label{txx}
\end{align}
the latest time of the triggers of all its in-neighbours agents.
For agent $v_{p}$, in order to determine the in-neighbour agent $v_i$'s state at time $t$ by equation (\ref{xpin}) replacing $p$ by $i$, the information of $t^{*}_{k_{i}(t)}$, $x_{i}(t^{*}_{k_{i}(t)})$ and $\sum_{j=1}^{m}L_{ij}x_j(t^{j}_{k_{j}(t)})$ from agent $v_i$ are required.

At any fixed time $s$, given $t^{p}_1,\cdots,t^{p}_l=t^{p}_{k_{p}(s)}$, Theorem \ref{coro1.1rr} implies that solving the following maximization problems
\begin{align}
&t_{l+1}^{p}=\max\Big\{\tau\ge s:~|x_{p}(r)-x_{p}(t^{p}_l)|\nonumber\\
&\le \sqrt{\frac{\gamma_{i}}{4|L_{ii}|}q_{p}(t^{*}_{k_{p}(s)})},\forall r\in[t^{*}_{k_{p}(s)},\tau]\Big\},\label{event5.1}
\end{align}
\begin{align}
\tau^{p}_{l+1}=\max\Big\{&\tau\ge s:
~|x_{p}(r)-x_{p}(t^{p}_l)|\nonumber\\&\le\delta_p(r),~\forall s\in[t^{*}_{k_{p}(s)},\tau]\Big\},
\label{event5.1e}
\end{align}
and
\begin{align}
&\tau^{p}_{l+1}=\max\Big\{\tau\ge s:
~|x_{p}(r)-x_{p}(t^{p}_l)|\nonumber\\&\le\sqrt{\frac{\gamma_p}{4|L_{pp}|}\hat{q}_{p}(t^{*}_{k_{p}(s)})},~\forall r\in[t^{*}_{k_{p}(s)},\tau]\Big\}
\label{event5.1rr}
\end{align}
can predict the next trigger time $t^{p}_{l+1}$.
Thus, the following algorithms are proposed.

{\em Distributed self-triggered principle $1$:} For agent $v_p$, pick $0<\gamma_p<1$, $\phi_p>0$ and $\alpha_p>0$.
\begin{enumerate}
\item At time $s=t^{p}_{l}$, search $\tau^{p}_{l+1}$ by the rule (\ref{event5.1}) or (\ref{event5.1e});
\item In case that none of agent $v_{p}$'s in-neighbours sends information to agent $v_{p}$ during $(t^{*}_{k_{p}(s)},\tau^{p}_{l+1})$, then $v_{p}$ triggers at time $t^{p}_{l+1}=\tau^{p}_{l+1}$;
\item In case that there exists $v_{p}$'s in-neighbour agents sending information at time $t\in (t^{*}_{k_{p}(s)},\tau^{p}_{l+1})$, then update time $s=t$ and $t^{*}_{k_{p}(s)}$ in (\ref{txx}) and go to step (1);
\end{enumerate}

{\em Distributed self-triggered principle $2$:} For agent $v_p$, pick $\gamma_{p}\in(0,1)$  and $0<\varepsilon_p<\frac{\sqrt{\gamma_p}}{2|L_{pp}|}$.
\begin{enumerate}
\item At time $s=t^{p}_{l}$, check whether $\sum_{i=1}^{m}L_{pj}\hat{x}_{j}(s)=0$ or not;
\item If true, then agent $v_p$ does not trigger and its state remains constant until one of its in-neighbors triggers at time $t>t^{*}_{k_{p}(s)}$, then update time $s=t$ and $t^{*}_{k_{p}(s)}$ in (\ref{txx}) and go to step (1);
\item If not true,then  at time $s$, search $\tau^{p}_{l+1}$ by (\ref{event5.1rr});
\item In case that there is no trigger occur in all $v_{p}$'s in-neighbours during $(t^{*}_{k_{p}(s)},\tau^{p}_{l+1})$, then $v_{p}$ triggers at time $t^{p}_{l+1}=\tau^{p}_{l+1}$;
\item In case that there exists one in-neighbour agent of $v_{p}$ triggers at time $t\in (t^{*}_{k_{p}(s)},\tau^{p}_{l+1})$: if $t-t^{p}_{l}\ge\varepsilon_p$ and $g_{p}(e_{p}(t))<0$, then then update time $s=t$ and $t^{*}_{k_{p}(s)}$, and go to step (1); otherwise $v_{p}$ triggers at time $t^{p}_{l+1}=t$.
\end{enumerate}

\begin{theorem}\label{thm5.1}
Suppose that $\mathcal G$ has spanning tree and $L$ is in the form of (\ref{PF}). Under the distributed self-triggered principle $1$ and $2$, system (\ref{mg}) reaches consensus exponentially and the Zeno behaviour can be excluded for principle $1$ with the prediction rule (\ref{event5.1e}) and principle $2$. In addition, in all cases,
$\lim_{t\to\infty}x_{i}(t)=\sum_{p=1}^{n_{K}}\xi^{K}_{p}x^{K}_{p}(0)$.
\end{theorem}
It can be seen that this theorem is a direct consequence from Theorem \ref{coro1.1rr}. By these principles,  the continuous monitoring of the system states is avoided.

\section{Examples}
In this section, a numerical example is given to demonstrate the present results.
Consider a network of four agents with the Laplacian matrix
\begin{eqnarray*}
L=\left[\begin{array}{rrrr}-7&3&0&4\\
1&-3&0&2\\
0&2&-7&5\\
0&0&4&-4
\end{array}\right]
\end{eqnarray*}
which corresponds a directed strongly connected weighted network. The initial value of each agent is randomly selected within the interval $[-5,5]$ in the simulation. Here, for this time, we take the initial values as $[2.5320,4.7160,-4.1310, 1.2830]^{\top}$. Figure \ref{fig:2} shows the evolution of the Lyapunov function $V(t)$ under the self-triggered principles provided in (\ref{event5.1}) with $\gamma_i=0.9$, in (\ref{event5.1e}) with $\phi_i=7$ and $\alpha_i=4>\lambda_2/(2\mu_m)=0.4035$, and in (\ref{event5.1rr}) with $\gamma_i=0.9$ and $\varepsilon_i=\sqrt{\gamma_i}/(4|L_{ii}|)$, for all $i\in\mathcal{I}$, in comparison to the original continuous feedback consensus protocol: $\dot{x}(t)=Lx(t)$. It can be seen that under all self-triggering principles, $V(t)$ approaches 0 exponentially, i.e., each agent reaches the consensus value $\bar{x}(0)=0.4814$ exponentially. Moreover, from Figure \ref{fig:2}, one can see that under (\ref{event5.1rr}) and (\ref{event5.1}) $V(t)$ converges at least as fast as under continuous feedback consensus protocol. Figure \ref{fig:3} illustrates the trigger times of each agents. It is shown that the inter-event times of each agent under the self-triggered principles by (\ref{event5.1e}) and (\ref{event5.1rr}) are lower-bounded by some positive constants. There are clearly less triggers in these two principles than that under the self-triggered principle provided in (\ref{event5.1}).


\begin{figure}[hbt]
\centering
\includegraphics[width=0.45\textwidth]{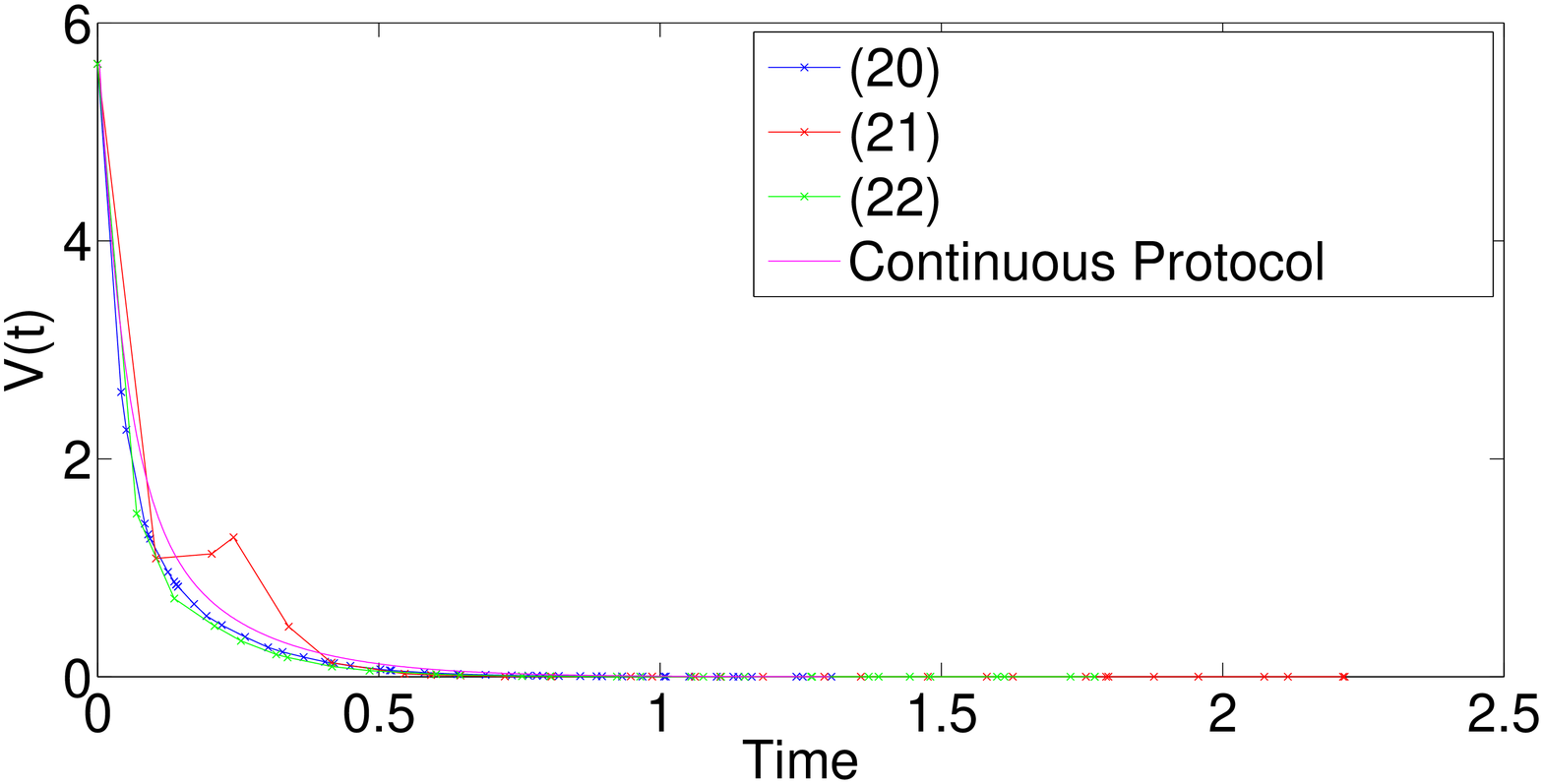}
\caption{Dynamics of $V(t)$ under distributed self-triggered principles by the prediction rules (\ref{event5.1}), (\ref{event5.1e}) and (\ref{event5.1rr}), compared with the continuous protocol $\dot{x}=Lx$. }
\label{fig:2}
\end{figure}

\begin{figure}[hbt]
\centering
\includegraphics[width=0.45\textwidth,height=0.3\textwidth]{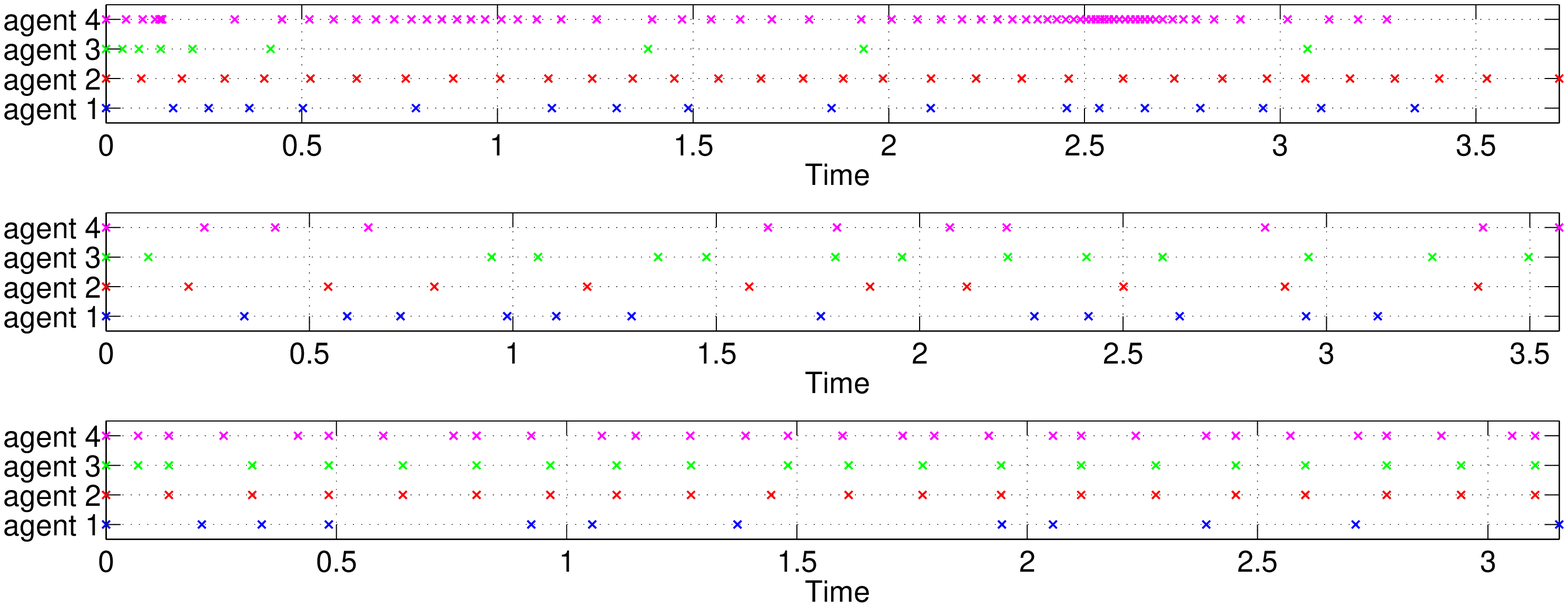}
\caption{The trigger times of each agents under distributed self-triggered principles by the prediction rules (\ref{event5.1}), (\ref{event5.1e}) and (\ref{event5.1rr}).}
\label{fig:3}
\end{figure}

\section{Conclusion}
In conclusion, we presented event-triggered and self-triggered principles in distributed
formulation for multi-agent systems with directed and possibly reducible topologies. We firstly  considered the case of directed strongly connected graph and then extended to reducible graph. We prove that if the underlying graph has a spanning tree, then the multi-agent reaches consensus exponentially by these principles, and Zeno behavior can be excluded in the cases. Then, we proposed self-triggered principles, which predicts the next triggering time instead of continuous monitoring the system states. The effectiveness the theoretical results is verified by the numerical example.



\end{document}